\documentclass[twoside]{article}

\usepackage[accepted]{aistats2021}

\usepackage[utf8]{inputenc} 
\usepackage[T1]{fontenc}    
\usepackage{hyperref}       
\usepackage{url}            
\usepackage{booktabs}       
\usepackage{amsfonts}       
\usepackage{nicefrac}       
\usepackage{microtype}      
\usepackage{caption}
\usepackage{wrapfig}

\usepackage{tabu}
\usepackage{graphicx}
\usepackage{subfigure}
\usepackage{xcolor,xspace}

\usepackage{titlesec}

\usepackage{caption}
\usepackage{amsmath,amstext, amsthm,amssymb}

\usepackage{enumitem}
\usepackage{calc}
\usepackage{dirtytalk}
\usepackage{mathtools}
\usepackage{thmtools,thm-restate}
\usepackage[textsize=tiny]{todonotes}
\usepackage{xspace}
\usepackage{tikz}
\usetikzlibrary{matrix,shapes,arrows,positioning,chains}

\usepackage{algorithm}
\usepackage{algorithmic}
\usepackage[algo2e,linesnumbered,ruled]{algorithm2e}

\usepackage{colortbl} 
\usepackage{xcolor} 
\usepackage{xfrac}
\usepackage{natbib}

\newcommand{\absl}[1]{\left \lvert #1 \right \rvert}
\newcommand{\asplusplus}{\textsc{rand-sampling}\xspace}

\newcommand{\bigo}[1]{\mathcal{O}\left ( #1 \right )}

\newcommand{\binarysearch}{\textsc{binary-search}\xspace}

\newcommand{\expect}[2]{\mathbb{E}_{#1}  \left[ \ #2 \ \right]}
\newcommand{\fantom}{\mbox{\textsc{fantom}\xspace}}
\newcommand{\greedy}{\textsc{repeatedGreedy}\xspace}
\newcommand{\samplegreedy}{\textsc{sampleGreedy}\xspace}
\newcommand{\fastSGS}{\textsc{FastSGS}\xspace}
\newcommand{\ignore}[1]{}
\newcommand{\ind}[1]{\mathcal{I} #1}
\newcommand{\kas}{\textsc{rep-sampling}\xspace}

\newcommand{\opt}{\textsc{opt}\xspace}

\newcommand{\pr}[2]{\mathsf{Pr}_{#1} \left ( #2 \right )}

\newcommand{\rndseq}{\textsc{rand-sequence}\xspace}
\newcommand{\simGreedy}{\textsc{simultaneousGreedys}\xspace}

\newcommand{\unif}{\textsc{unif-sampling}\xspace}
\DeclareMathOperator*{\argmax}{arg\,max}

\newtheorem{problem}{Problem}
\newtheorem{proposition}{Proposition}

\newtheorem{definition}{Definition}
\newtheorem{theorem}{Theorem}
\newtheorem{lemma}{Lemma}

\SetKwInput{KwInput}{Input}
\SetKwInput{KwOutput}{Output}

\begin{document}

\twocolumn[

 \runningtitle{Adaptive Sampling for Fast Constrained Maximization of Submodular Functions}

\aistatstitle{Adaptive Sampling for Fast Constrained Maximization\\of Submodular Functions}

\aistatsauthor{Francesco Quinzan \And Vanja Dosko\v{c} \And Andreas G\"{o}bel \And Tobias Friedrich}

\aistatsaddress{Hasso Plattner Institute\\Potsdam, Germany \And Hasso Plattner Institute\\Potsdam, Germany \And Hasso Plattner Institute\\Potsdam, Germany \And Hasso Plattner Institute\\Potsdam, Germany} ]

\begin{abstract}

Several large-scale machine learning tasks, such as data summarization, can be approached by maximizing functions that satisfy submodularity. These optimization problems often involve complex side constraints, imposed by the underlying application. In this paper, we develop an algorithm with poly-logarithmic adaptivity for non-monotone submodular maximization under general side constraints. The adaptive complexity of a problem is the minimal number of sequential rounds required to achieve the objective.\\
Our algorithm is suitable to maximize a non-monotone submodular function under a $p$-system side constraint, and it achieves a $(p + \bigo{\sqrt{p}})$-approximation for this problem, after only poly-logarithmic adaptive rounds and polynomial queries to the valuation oracle function. Furthermore, our algorithm achieves a $(p + \bigo{1})$-approximation when the given side constraint is a $p$-extendible system.\\
This algorithm yields an exponential speed-up, with respect to the adaptivity, over any other known constant-factor approximation algorithm for this problem. It also competes with previous known results in terms of the query complexity. We perform various experiments on various real-world applications. We find that, in comparison with commonly used heuristics, our algorithm performs better on these instances.
\end{abstract}

\section{Introduction}

Several machine learning optimization problems consist of maximizing submodular functions. Examples include subset selection~\citep{DasK18}, data summarization~\citep{Lin:2010,DBLP:conf/icml/MirzasoleimanBK16}, and Bayesian experimental design~\citep{Chaloner95bayesianexperimental,DBLP:journals/jmlr/KrauseSG08}. These problems often involve constraints imposed by the underlying application. For instance, in video summarization tasks several constraints on the solution space arise based on qualitative features and contextual information~\citep{DBLP:conf/icml/MirzasoleimanBK16}.

The problem of maximizing a submodular function is NP-hard~\citep{DBLP:journals/jacm/Feige98}. However, several approximation algorithms for this problem have been discovered over the years. For monotone submodular functions, the classical result of~\citet{DBLP:journals/mp/NemhauserWF78} shows that  a simple greedy algorithm provides a $(1 - 1/e)$-approximation guarantee for the maximization of a \emph{monotone} submodular function under a uniform constraint. If an additional matroid constraint is imposed on the solution space, then greedy achieves a $(1/2)$-approximation guarantee on this problem~\citep{inbook}. A constant-factor approximation guarantee can also be achieved in the case of a knapsack constraint~\citep{DBLP:journals/orl/Sviridenko04}.\\
More complex constraints require more complex heuristics. Several algorithms have been discovered, to maximize a monotone submodular function under general side constraints such as $p$-systems and multiple knapsacks~\citep{DBLP:conf/soda/BadanidiyuruV14,DBLP:conf/focs/ChekuriP05}. These algorithms include streaming algorithms~\citep{DBLP:conf/kdd/BadanidiyuruMKK14,DBLP:conf/icalp/ChekuriGQ15,DBLP:journals/mp/ChakrabartiK15}, centralized algorithms~\citep{DBLP:conf/soda/BadanidiyuruV14,DBLP:conf/aaai/MirzasoleimanBK15}, and distributed algorithms~\citep{DBLP:conf/nips/MirzasoleimanKSK13,DBLP:journals/topc/KumarMVV15}.

\begin{table*}[t!]
\caption{Results for non-monotone submodular maximization with a $p$-system side constraint. Here, $n$ is the problem size, $r$ is the maximum size of a feasible solution, and $p$ the the parameter for the side constraint. The results on the adaptivity for previously known algorithms follow from the adaptivity of the greedy algorithm. Note also that all bounds on the adaptivity and query complexity for $p$-systems are parameterized by $p$. Whether it is possible to obtain bounds independent of $p$ for this problem remains an open question.}
\label{table:results}
 \begin{center}
 \scalebox{0.95}{
\begin{tabular}{c|c|c|c|} 
 \cline{2-4} 
 & \multicolumn{3}{c|}{\textbf{$p$-systems}} \\
 \hline
  \multicolumn{1}{|c|}{\textbf{Algorithm}}
  & \textbf{Approx.} &\textbf{Adaptivity} & \textbf{Query Complexity} \\
   \hline
 \multicolumn{1}{|c|}{$
 \begin{array}{c}
 \mbox{\kas{}}\\
 \mbox{[this work]}\\
 \end{array}
 $}
 & $\approx p + \bigo{\sqrt{p}}$ & $\bigo{\sqrt{p}  \log n \log \frac{r}{p } \log r }$ & $\bigo{\sqrt{p} n \log n \log \frac{r}{p } \log r }$ \\
    \hline
 \multicolumn{1}{|c|}{$
 \begin{array}{c}
 \mbox{\fastSGS{}}\\
 \mbox{\citep{DBLP:journals/corr/abs-2009-13998}}\\
 \end{array}
 $}
 & $\approx p + \bigo{\sqrt{p}}$ & $\bigo{p n \log n}$ & $ \bigo{p n \log n}.$ \\
 \hline
  \multicolumn{1}{|c|}{$
 \begin{array}{c}
 \mbox{\simGreedy{}}\\
 \mbox{\citep{DBLP:journals/corr/abs-2009-13998}}\\
 \end{array}
 $}
 & $ p + \bigo{\sqrt{p}}$ & $\bigo{\sqrt{p}  r }$ & $\bigo{p r n }$ \\
 \hline
  \multicolumn{1}{|c|}{$
 \begin{array}{c}
 \mbox{\greedy{}}\\
 \mbox{\citep{DBLP:conf/colt/FeldmanHK17} }\\
 \end{array}
 $}
  & $p + \bigo{\sqrt{p}}$ & $\bigo{\sqrt{p} r}$ & $\bigo{\sqrt{p} nr}$ \\
  \hline
    \multicolumn{1}{|c|}{$
 \begin{array}{c}
 \mbox{\fantom{}}\\
 \mbox{\citep{DBLP:conf/icml/MirzasoleimanBK16} }\\
 \end{array}
 $}
  & $\approx 2 p$ & $\bigo{p r}$ & $\bigo{p nr}$ \\
 \hline
   \multicolumn{1}{|c|}{$
 \begin{array}{c}
 \mbox{\greedy{}}\\
 \mbox{\citep{DBLP:conf/wine/GuptaRST10} }\\
 \end{array}
 $}
 & $\approx 3 p$ & $\bigo{p r}$ & $\bigo{p nr}$ \\ \hline
\multicolumn{4}{c}{}\\
 \cline{2-4} 
 & \multicolumn{3}{c|}{\textbf{$p$-extendable systems}} \\
 \hline
 \multicolumn{1}{|c|}{$
 \begin{array}{c}
 \mbox{\kas{}}\\
 \mbox{[this work]}\\
 \end{array}
 $}
 & $\approx p + \bigo{1}$ & $\bigo{  \log n \log^2 r }$ & $\bigo{n \log n \log^2 r } $ \\
 \hline
  \multicolumn{1}{|c|}{$
 \begin{array}{c}
 \mbox{\fastSGS{}}\\
 \mbox{\citep{DBLP:journals/corr/abs-2009-13998}}\\
 \end{array}
 $}
 & $\approx p + \bigo{1}$ & $\bigo{p^2 n \log n}$ & $ \bigo{p^2 n \log n}.$ \\
 \hline
  \multicolumn{1}{|c|}{$
 \begin{array}{c}
 \mbox{\simGreedy{}}\\
 \mbox{\citep{DBLP:journals/corr/abs-2009-13998}}\\
 \end{array}
 $}
 & $ p + \bigo{1}$ & $\bigo{p  r }$ & $\bigo{p^2 r n }$ \\
 \hline
  \multicolumn{1}{|c|}{$
 \begin{array}{c}
 \mbox{\samplegreedy{}}\\
 \mbox{\citep{DBLP:conf/colt/FeldmanHK17} }\\
 \end{array}
 $}
  & $p + \bigo{1}$ & $\bigo{ r}$ & $\bigo{ nr}$ \\
  \hline
  \multicolumn{4}{c}{}\\
 \cline{2-4} 
 & \multicolumn{3}{c|}{\textbf{$p$-matchoids}} \\
 \hline
 \multicolumn{1}{|c|}{$
 \begin{array}{c}
 \mbox{Parallel Greedy$^{\dagger}$ }\\
 \mbox{\citep{DBLP:conf/stoc/ChekuriQ19} }\\
 \end{array}
 $}
  & $\approx \frac{4 p + 4}{1 + o(1)}$ & $\bigo{\log n \log r}$ & $\bigo{n \log n \log r}$ \\
 \hline
\end{tabular}
}
\end{center}
{
\footnotesize $ $\vspace{5pt}\\$^{\dagger}$The Parallel Greedy algorithm requires access to the rank oracle for the underlying $p$-matchoid system. This oracle is strictly less general then the independence oracle required by all other algorithms in Table \ref{table:results}.}
\end{table*}
Many algorithms have also been proposed, to maximize \emph{non-monotone} submodular functions under a variety of constraints~\citep{DBLP:conf/icalp/FeldmanNS11,DBLP:journals/siamcomp/ChekuriVZ14,DBLP:conf/wine/GuptaRST10,DBLP:conf/stoc/LeeMNS09,DBLP:journals/siamcomp/FeigeMV11,DBLP:journals/siamcomp/BuchbinderFNS15}. These algorithms yield good approximation guarantees, but their run time is polynomial in the number of data-points, and polynomial in the number of additional side constraints.\\
Recently, algorithms were discovered to maximize a non-monotone submodular function under very general side constraints \citep{DBLP:conf/icml/MirzasoleimanBK16,DBLP:conf/colt/FeldmanHK17}. These constant-factor approximation algorithms scale polynomially in the number of data-points, but also in the number of additional side constraints.

In some cases, approximation algorithms do not exhibit increasingly worse run time in the number of constraints. This is the case when maximizing a submodular function under \emph{$p$-extendible systems} or \emph{$p$-matchoid} side constraints \citep{DBLP:conf/colt/FeldmanHK17,DBLP:conf/stoc/ChekuriQ19}. These side constraints are strictly less general than those studied in \citet{DBLP:conf/icml/MirzasoleimanBK16}, but they are general enough to capture a variety of interesting applications.

Submodular functions are learnable in the standard $\mathsf{PAC}$ and $\mathsf{PMAC}$ models \citep{DBLP:journals/cacm/Valiant84,DBLP:conf/stoc/BalcanH11}: given a collection of sampled sets and their submodular function values, it is possible to produce a surrogate that mimics the behavior of that function, on samples drawn from the same distribution. However, submodular objectives cannot be optimized from the training data we use to learn them \citep{DBLP:conf/stoc/BalkanskiRS17,DBLP:conf/icml/RosenfeldBGS18}. The reason is that, when learning from samples, resulting surrogate functions can be inapproximable, and their global optima can be far away from the true optimum.\\
Using an adaptive sampling framework \citep{doi:10.1080/01621459.1990.10474975}, it is possible to design algorithms that reach a constant-factor approximation guarantee in poly-logarithmic adaptive rounds for submodular maximization, both in the monotone \citep{DBLP:conf/icml/BalkanskiS18,DBLP:conf/stoc/BalkanskiRS19,DBLP:conf/stoc/BalkanskiS18,DBLP:conf/soda/FahrbachMZ19,DBLP:conf/soda/EneN19} and  non-monotone case \citep{DBLP:conf/nips/BalkanskiBS18,DBLP:conf/stoc/EneNV19,DBLP:conf/icml/FahrbachMZ19}. At each adaptive round, calls to the value oracle function are queried independently. However, lower-bounds for algorithms with low adaptivity are also known \citep{DBLP:journals/corr/abs-2002-09130,DBLP:conf/stoc/BalkanskiS18}.
\paragraph{Our contribution.} Focusing on sampling techniques, we study the problem of maximizing a non-monotone submodular function, to which we have oracle access. Furthermore, we consider general $p$-system and $p$-extendible system side constraints for this problem. 

Our algorithm has access to the side constraint structure via an oracle. Standard oracle models in the literature are: the \emph{independence oracle}, which takes as input a set and returns whether that set is a feasible solution; the \emph{rank oracle}, that returns the maximum cardinality of any feasible solution contained in a given input set; and the \emph{span oracle}, which for an input set $S$ and a point $\{e\}$ it returns whether or not $S \cup \{e\}$ has a higher rank than $S$. In this work, we assume access to the independence oracle, which is the most general oracle model of the three.

In this work, we develop the first algorithm with poly-logarithmic adaptivity suitable to maximize a non-monotone submodular function under a $p$-system side constraint and a $p$-extendible system. In contrast to all previous algorithms with low adaptivity, our algorithm only requires access to the independence oracle for the side constraints. This algorithm achieves strong approximation guarantees and run time, competing with known algorithms for this problem (see Table~\ref{table:results}).\\
We study the performance of our algorithm in two real-world applications, video summarization and Bayesian experimental design. We test our algorithm against other commonly used heuristics for this problem, and show that our algorithm comes out on top.

Our paper is organized as follows. We define the problem in Section \ref{sec:problem_description}, and we describe our algorithm in Section \ref{sec:algorithms}. Our theoretical analysis is presented in Sections \ref{sec:theoretical_analysis}-\ref{sec:ind_oracle}. Applications and experiments are discussed in Sections \ref{experimental:framework}-\ref{sec:bayesian_design}. We conclude in Section \ref{sec:conclusion}.
%
%
\section{Problem Description}
\label{sec:problem_description}
%
%
\paragraph{Submodularity.} In this paper, we study optimization problems that can be approached by maximizing an oracle function that, given a solution set, estimates its quality. We assume that oracle functions are \emph{submodular}.
\begin{definition}[Submodularity]
\label{def:submodular}
Given a finite set $V$, we call a set function $f\colon 2^V\to \mathbb{R}_{\geq 0}$ \emph{submodular} if for all $S,U\subseteq V$ we have that $f(S) + f(U) \geq f(S \cup U) + f(S \cap U)$.
\end{definition}
Note that we only consider functions that do not attain negative values. This is because submodular functions with negative values cannot be maximized, even approximately (see \citet{DBLP:journals/siamcomp/FeigeMV11}).
%
%
\paragraph{$p$-Systems.} We study the problem of maximizing a submodular function under additional side constraints, defined as a $p$-system side constraint. As discussed, i.e., in \citet{DBLP:conf/icml/MirzasoleimanBK16,DBLP:conf/wine/GuptaRST10}, these constraints are significantly more general than standard matroid intersections, and they arise in various domains, such as movie recommendation, video summarization, and revenue maximization.

Given a collection of feasible solutions $\ind{}$ over a ground set $V$ and a set $T \subseteq V$, we denote with $\ind{}\mid_T$ a collection consisting of all sets $S \subseteq T$ that are feasible in $\ind{}$. Furthermore, a base for $\ind{}$ is any maximum feasible set $U \in \ind{}$. We define $p$-systems as follows.
\begin{definition}
A $p$-system $\ind{}$ over a ground set $V$ is a collection of subsets of $V$ fulfilling the following three axioms:
\begin{itemize}
    \item $\emptyset \in \ind{}$;
    \item for any two sets $S \subseteq \Omega \subseteq V$, if $\Omega \in \ind{}$ then $S \in \ind{}$;
    \item for any set $T \subseteq V$ and any bases $S, U \in \ind{}\mid_T$ it holds $\absl{S} \leq p \absl{U}$.
\end{itemize}
\end{definition}
The second defining axiom is referred to as subset-closure or downward-closed property. With this notation, we study the following problem.
\begin{problem}
\label{problem}
Given a submodular function $f:2^{V} \rightarrow \mathbb{R}_{\geq 0}$ and a $p$-system $\ind{}$, find a set $S \subseteq V$ maximizing $f(S)$ such that $S \in \ind{}$.
\end{problem}
%
%
\paragraph{$p$-extendible Systems.} We also consider a family of side constraints of intermediate generality, commonly referred to as $p$-extendible systems. These side constraints are strictly less general than $p$-systems, but they capture various types of constraints found in practical applications.

Our main motivation in studying these constraints is that they admit algorithms that obtain strong approximation guarantees, in much less time than in the case of the $p$-systems. Hence, algorithms for $p$-extendible systems scale much better than for general $p$-systems.

These $p$-extendible systems were first studied by \cite{DBLP:conf/esa/Mestre06}, and they are defined as follows.
\begin{definition}
A $p$-extendible system $\ind{}$ over a ground set $V$ is a $p$-system, that fulfills the following additional axiom: for every pair of sets $S, \Omega \in \ind{}$ with $S \subset \Omega$, and for every element $e \notin S$, there exists a set $U \subseteq \Omega \setminus S$ of size $\absl{U} \leq p$ such that $\Omega \setminus U \cup \{ e\} \in \ind{}$.
\end{definition}
These side constraints generalize matroid intersections and $p$-matchoids. While being strictly less general than $p$-systems, this definition captures many interesting constraints, such as the intersection of matroids \citep{DBLP:conf/esa/Mestre06}. In this paper, we also study the following problem.
\begin{problem}
\label{problem2}
Given a submodular function $f:2^{V} \rightarrow \mathbb{R}_{\geq 0}$ and a $p$-extendible system $\ind{}$, find a set $S \subseteq V$ maximizing $f(S)$ such that $S \in \ind{}$.
\end{problem}
%
%
\paragraph{Adaptivity.} An algorithm is $T$-adaptive if every query $f(S)$ for the $f$-value of a solution $S$ occurs at a round $i \in [T]$ such that $S$ is independent of the values $f(S')$ of all other queries at round $i$, with at most polynomial queries at each round in the problem size. The query complexity is the number of calls to the evaluation oracle function.
%
%
\begin{algorithm2e}[t]
	\caption{\rndseq{}$(X, S, \ind{})$}
	\label{alg:rndseq}
	\begin{algorithmic}[1]
	    \STATE $A \gets \emptyset$;
 	    \WHILE{$X \neq \emptyset $}\label{alg:beginWhileAlg3}
 	    \STATE sort the points $\{x_i\}_i = X$ randomly;
 	    \STATE $\eta \gets \max \{ j \colon S \cup A \cup \{x_i\}_{i \leq j} \in \ind{} \} $;
 	    \STATE $A \gets A \cup \{x_1, \dots,x_{\eta}\}$;\label{alg:sequence:alg}
 	    \STATE $X \gets \{e \in X\setminus (S \cup A) \colon S \cup A \cup e  \in \ind{} \}$;\label{alg:rndseq:Y1}
        \ENDWHILE \label{alg:endWhileAlg3}
   	\STATE \textbf{return} $A$;
   	\end{algorithmic}
\end{algorithm2e}
\paragraph{Notation.}
For any submodular evaluation oracle function $f\colon 2^V \rightarrow \mathbb{R}$ and sets $S, U \subseteq V$, we define the marginal value of $S$ with respect to $U$ as $f (U \mid S) = f(S\cup U) - f(S)$.\\
Throughout the paper, we always use the notation introduced in Problem \ref{problem}: we denote with $f$ the evaluation oracle function, with $V$ the ground set, and with $\ind{}$ the $p$-system side constraint. We denote with $\opt$ a solution to Problem \ref{problem}, and we denote with $n$ the size of the ground set $V$, i.e., $n$ is the number of singletons in our solution space. We also denote with $r$ the maximum size of a feasible solution.\\
The notation introduced in Algorithm \ref{alg:rndseq}-\ref{alg:ksampl} is used consistently throughout the paper.
%
%
\section{Algorithms}
\label{sec:algorithms}

Our method consists of three parts (see Algorithms~\ref{alg:rndseq}-\ref{alg:ksampl}). We call these algorithms  \rndseq{}, \asplusplus{}, and \kas{} respectively. These algorithms also call the \binarysearch{} and \unif{} sub-routines. The following is a description of each algorithm and sub-routine.

\paragraph{\rndseq{}.} It is based on the work of \citet{DBLP:journals/jcss/KarpUW88}. Given as input a ground set $X$, a current solution $S$, and a $p$-system $\ind{}$, this algorithm finds a random set $A$ such that $S \cup A$ is a base for $\ind{}$.
\begin{algorithm2e}[t]
	\caption{\asplusplus{}$(f, V, \ind{},\lambda, \varepsilon, \varphi_1)$.}
	\label{alg:samplplusplus}
	\begin{algorithmic}[1]
	\STATE $S \gets \emptyset$;
	\STATE $X \gets \argmax_e \{ f(e) \colon e \in V \land e \in \ind \}$;
	\STATE $\delta \gets f(X)$, $\delta_0 \gets \lambda f(X) $;
	\WHILE{$\delta \geq \delta_0$}\label{alg:samplplusplus:beginWhile1}
   \WHILE {$X \neq \emptyset$}\label{alg:sampl:beginWhile2}
    \STATE $\{a_j\}_{j\in J } \gets \mbox{\rndseq{}}(X,S,\ind{})$;\label{alg:sampl:binaryBegin}
    \STATE $\eta \gets \mbox{\binarysearch}(J, \min \{j \in J \colon \absl{X_j} < (1 - \varepsilon ) \absl{X} \})$ with $X_j = \{e \in X \colon f (e \mid S \cup \{a_1, \dots, a_{j-1} \} ) \geq \delta \land S \cup \{a_1, \dots, a_{j-1} \} \cup e \in \ind \}$
   	\STATE $A \gets \mbox{\unif}(\{a_1, \dots, a_{\eta- 1}\}, \varphi_1)$;
    \STATE $X \gets X_{\eta }$; $S \gets S \cup A $;
    \ENDWHILE \label{alg:sampl:endWhile2} 
    \STATE $\delta \gets ( 1 - \varepsilon )\delta$;
    \STATE $X \gets \{e \in V \colon f(e\mid S) \geq \delta \land S \cup e \in \ind  \}$;\label{alg:sampl:defX}
    \ENDWHILE \label{alg:samplplusplus:endWhile1}
   	\STATE \textbf{return} $S$;
   	\end{algorithmic}
\end{algorithm2e}
\paragraph{\asplusplus{}.} This algorithms generalizes a sampling algorithm proposed in \citet{DBLP:conf/stoc/BalkanskiRS19} to non-monotone submodular maximization. This algorithm requires as input an oracle function $f$, a ground set $V$, a $p$-system or $p$-extendible system $\ind{}$, and parameters $\lambda, \varepsilon,  \varphi_1$. The parameter $\lambda$ determines the total number of iterations for the \asplusplus{}, the parameter $\varepsilon$ determines the rate with which the variable $\delta$ decreases, whereas $ \varphi_1$ determines the distribution for the \unif{} sub-routine. For a constant $\delta $, points are added to the current solution yielding a marginal contribution upper-bounded by $\delta$. Note that at each adaptive step, the \asplusplus{} uses the \binarysearch{} and the \unif{} sub-routine. If Algorithm \ref{alg:samplplusplus} reaches an iteration with $X = \emptyset$, then it decreases the value of $\delta $ so that points with lower marginal contribution can be added to the current solution.
\begin{algorithm2e}[t]
	\caption{\kas{}$(f, V, \ind{}, \varepsilon, \varphi_1, \varphi_2, m)$}
	\label{alg:ksampl}
	\begin{algorithmic}[1]
	\STATE $\lambda \gets \varepsilon (p + 1)/m $;
	\FOR{$j \leq m$ iterations} \label{for:kas_begin}
    {
        \STATE $\Omega_j \gets \asplusplus{}(f, V, \ind{}, \lambda, \varepsilon, \varphi_1)$;\label{for:kas_line1}
        \STATE $\Lambda_j \gets \mbox{\unif} (\Omega_j, \varphi_2)$;\label{alg:k-sampl_random}
        \STATE $V \gets V \setminus \Omega_j$;
   	}
   	\ENDFOR  \label{for:kas_end}
   	\STATE \textbf{return} $\argmax_j \{f(\Omega_j), f(\Lambda_j) \}$\;
   	\end{algorithmic}
\end{algorithm2e}
\paragraph{\binarysearch.} This sub-routine is just the standard binary search algorithm. It is used to locate an index $\eta$ such that $\eta = \min \{ j \colon \absl{X_j} \leq (1 - \varepsilon ) \absl{X} \}$, with $X_j $, where the index $j$ spans over the set $J$. This sub-routine uses the fact that, due to submodularity, it holds $\absl{X_j} \leq \absl{X_{j + 1}}$ for all $j \in J$.
\paragraph{\unif{}.} For a given input set and probability $\varphi$, this algorithm samples points of the input set independently, with probability $\varphi$.
\paragraph{\kas{}.} This algorithm requires as input an oracle function $f$ a ground set $V$, a $p$-system or $p$-extendible system $\ind{}$ and parameters $\lambda ,m, \varepsilon$ and $\varphi_1, \varphi_2$. At each step, the \kas{} calls Algorithm \ref{alg:samplplusplus} to find a partial solution $\Omega_j$. Then, Algorithm \ref{alg:ksampl} samples a subset of $\Omega_j$, where each point is drawn independently with probability $\varphi_2$. Afterwards, the \kas{} 
removes all points of $\Omega_j$ from the ground set, and it runs the \kas{} on the resulting ground set. This procedure is iterated $m $ times.
%
%
\section{Analysis for p-Systems}
\label{sec:theoretical_analysis}
In this section, we discuss theoretical run time analysis results for Problem \ref{problem}. We remark that \textbf{all proofs are deferred to the appendix}. Approximation guarantees for Algorithm \ref{alg:ksampl} follow from the following general theorem.
\begin{restatable}{theorem}{mainthm}
\label{thm1}
Fix constants $\varepsilon \in (0, 1)$, $m \geq 2$, $\varphi_1 = 1$, and $\varphi_2 = 1/2$. Denote with $\Omega^*$ the output of Algorithm \ref{alg:ksampl}. Then,
\begin{equation*}
    f(\opt) \leq m \left ( \frac{(1 + \varepsilon )(p + 1)}{(1 - \varepsilon)^2 (m - 1)} + 2 \right ) \expect{}{f(\Omega^*)}.
\end{equation*}
\end{restatable}
A proof of this theorem is given in Appendix \ref{appendix:A}-\ref{appendix:B1}.

We estimate the number of adaptive rounds until Algorithm~\ref{alg:ksampl} reaches the desired approximation guarantee. The following lemma holds.
\begin{restatable}{lemma}{runTime}
\label{lemma:runTime}
Fix constants $\varepsilon \in (0, 1)$, $\varphi_1, \varphi_2 \in [0, 1]$ and $m \geq 0 $. Then Algorithm \ref{alg:ksampl} terminates after {$\bigo{\frac{m}{\varepsilon^{2}} \log \left ( \frac{r}{p\varepsilon } \right ) \log r \log n }$} rounds of adaptivity. Furthermore, Algorithm \ref{alg:ksampl} has query complexity of {$\bigo{\frac{m n}{\varepsilon^{2}} \log \left ( \frac{r}{p\varepsilon } \right ) \log r \log n }$}.
\end{restatable}
A proof of this result is given in Appendix \ref{appendix:C}. The following lemma follows from Theorem \ref{thm1} and Lemma~\ref{lemma:runTime}. 
\begin{restatable}{lemma}{specific}
\label{thm:specific}
Fix a constant $\varepsilon \in (0, 1)$, and define parameters $m = 1 + \lceil \sqrt{(p + 1)/2} \rceil$, $\varphi_1 = 1$, and $\varphi_2 = 1/2$. Denote with $\Omega^*$ the optimal solution found by Algorithm \ref{alg:ksampl}. Then,
\begin{equation*}
    f(\opt) \leq \frac{1 - \varepsilon }{(1 - \varepsilon)^2 } \left (p + 2\sqrt{2(p + 1)} + 5 \right ) \expect{}{f(\Omega^*)}.
\end{equation*}
Furthermore, with this parameter choice Algorithm \ref{alg:ksampl} terminates after $\bigo{\frac{\sqrt{p}}{\varepsilon^{2}}  \log n \log \left (\frac{r}{p \varepsilon } \right )\log  r }$ rounds of adaptivity, and its query complexity is $\bigo{\frac{\sqrt{p} n}{\varepsilon^{2}} \log n \log \left (\frac{r}{p \varepsilon } \right )\log  r}$.
\end{restatable}
A proof is given in Appendix \ref{appendix:E}. We remark that there exists an algorithm with constant adaptivity for unconstrained non-monotone submodular maximization that achieves an approximation guarantee arbitrarily close to $1/2$ (see \citet{DBLP:conf/stoc/0003FK19}). Using this algorithm as a sub-routine in line \ref{alg:k-sampl_random} of Algorithm \ref{alg:ksampl} yields a constant-factor improvement over the approximation guarantee of Lemma \ref{thm:specific}, without affecting the upper-bound on the adaptivity. However, this algorithm requires access to a continuous extension of the value oracle $f$, whereas Algorithm \ref{alg:ksampl} only requires access to $f$.
%
%
\section{Analysis for p-extendible Systems}
In this section, we perform the theoretical analysis for the \kas{}, when maximizing a non-monotone submodular function under a $p$-extendible system side constraint, as in Problem \ref{problem2}. We prove that, with different sets of input parameters, our algorithm has adaptivity and query complexity that is not dependent on $p$. Again, all proofs are deferred to the appendix. The following theorem holds.
\begin{restatable}{lemma}{mainthmext}
\label{thm:extendible}
Fix parameters $\varepsilon \in (0, 1)$, $m = 1$, $\varphi_1 = (p + 1)^{-1}$, and $\varphi_2 \in [0, 1]$. Denote with $\Omega^*$ the output of Algorithm \ref{alg:ksampl}. Then,
\begin{equation*}
    f(\opt) \leq \frac{(1 + \varepsilon)(p + 1)^2}{p(1 - \varepsilon)^2} \expect{}{f(\Omega^*)}.
\end{equation*}
With this parameter choice, Algorithm \ref{alg:ksampl} terminates after $\bigo{\varepsilon^{-2} \log n \log \left (\frac{r}{\varepsilon } \right )\log  r}$ rounds of adaptivity, and it requires $\bigo{\frac{n}{\varepsilon^{2}} \log n \log \left (\frac{r}{\varepsilon } \right )\log  r}$ function evaluations.
\end{restatable}
For a proof of this result see Appendix \ref{appendix:G}.
The proof of this lemma is based on the work of \cite{DBLP:conf/colt/FeldmanHK17}, together with the fact that Algorithm \ref{alg:samplplusplus} yields expected marginal increase lower-bounded by the best possible greedy improvement, up to a multiplicative constant.

We remark that Lemma \ref{thm:extendible} also holds when side constraints are $p$-matchoids and the intersections of matroids, since $p$-extendible systems are a generalization of both.
%
%
\section{Query Complexity and Adaptivity of the Independence Oracle}
\label{sec:ind_oracle}
We conclude our analysis with a general discussion on the performance of Algorithm \ref{alg:ksampl} in the number of calls to the independence oracle for the $p$-system  constraint. The independence oracle takes as input a set $S$, and returns as output a Boolean value, true if the given set is independent in $\ind{}$ and false otherwise. The following lemma holds.
\begin{restatable}{lemma}{indOracle}
\label{lemma:indOracle}
Fix parameters $\varepsilon \in (0, 1)$, $m \geq 1$, and $\varphi_1, \varphi_2 \in [0, 1]$. Then Algorithm \ref{alg:ksampl} requires expected $\bigo{\frac{m \sqrt{n}}{\varepsilon^{2}} \log \left ( \frac{r}{p\varepsilon } \right ) \log r \log n }$ rounds of independent calls to the oracle for the $p$-system constraint. Furthermore, the total number of calls to the independence system is $\bigo{\frac{m n^{3/2}}{\varepsilon^{2}} \log \left ( \frac{r}{p\varepsilon } \right ) \log r \log n }$.
\end{restatable}
A proof of this result is given in Appendix \ref{appendix:D}, and it follows from the work of \citet{DBLP:journals/jcss/KarpUW88}. Note that the rounds of independent calls to the oracle are sub-linear, but not poly-logarithmic in the problem size. The reason is that Algorithm \ref{alg:rndseq} requires $\bigo{\sqrt{n}}$ rounds of independent calls to the oracle for the $p$-system. We are not aware of any algorithm that finds a base in less than $\bigo{\sqrt{n}}$ rounds. Furthermore, it is well-known that there is no algorithm that obtains an approximation guarantee that is constant in the problem size for Problem \ref{problem}, than $\tilde{\Omega}(n^{1/3})$ steps of independent calls to the oracle for the $p$-system constraint (see \citet{DBLP:journals/jcss/KarpUW88,DBLP:conf/stoc/BalkanskiRS19}).

For a $p$-system $\ind{}$, the rank of a set $S$ is the maximum cardinality of its intersection with a maximum independent set in $\ind{}$. Given access to an oracle that returns the rank of a set in~$\ind{}$, it is possible to design an algorithm that finds a maximum independent set of a $p$-system in $\bigo{\log n^2}$ rounds of independent calls to the rank oracle (see \citet{DBLP:journals/jcss/KarpUW88}). However, this work focuses on general constraints where the rank of a set is not known.

%
%
\begin{figure*}[t]
\includegraphics[width=\linewidth]{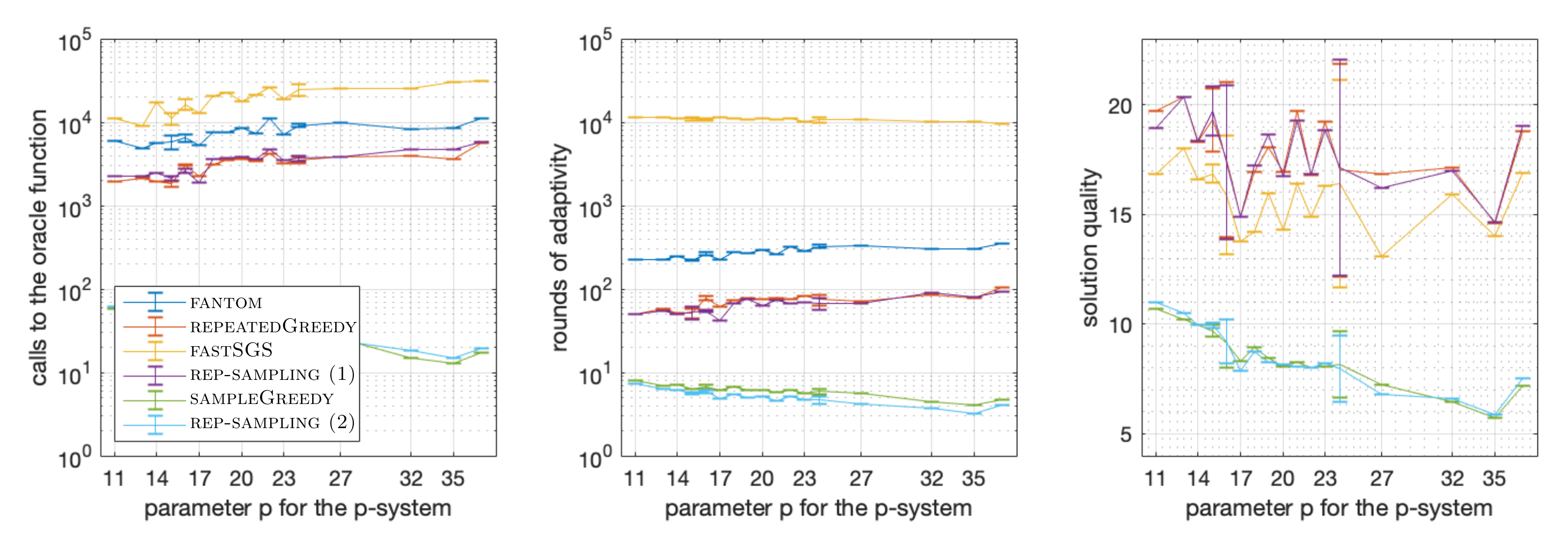}
\caption{Figure 1: Results for the experiments on Video Summarization on movie segments taken from FLIC \citep{modec13}. Each plot shows the average performance over segments with fixed $p$. Error bars correspond to the best and worst case. Note that the $y$-axis in the two leftmost plots uses a logarithmic scale. The \fastSGS{} uses parameters $\ell = \lfloor 2 + \sqrt{p + 1} \rfloor$ and $\varepsilon = 0.1$; the \kas{} \textsc{(1)} uses parameters $\varepsilon = 0.1$, $m = 1 + \lceil \sqrt{(p + 1)/2 }\rceil$, $\varphi_1 = 1$, $\varphi_2 = 0.5$;  the \kas{} \textsc{(2)} uses parameters $\varepsilon = 0.01$, $m = 1$, $\varphi_1 = (1 + p)^{-1}$, $\varphi_2 = 1$.}
\label{fig:1}
\end{figure*}
\begin{figure*}[t]
{\includegraphics[width=\linewidth]{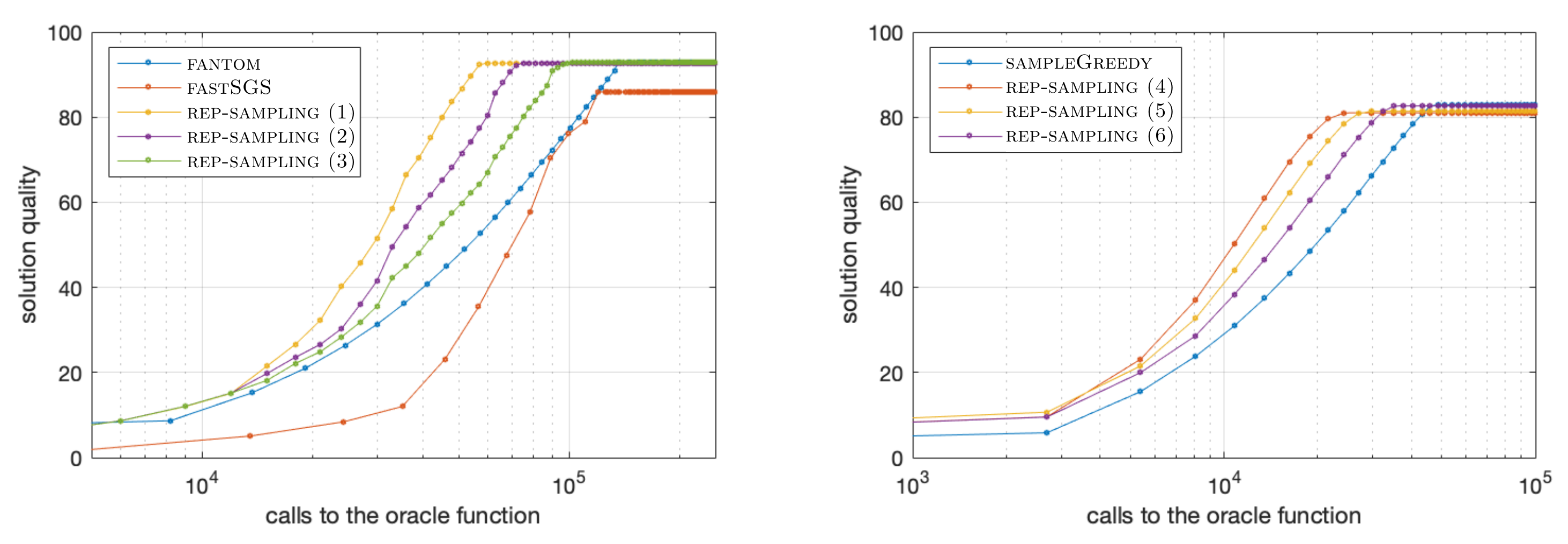}
\caption{Figure 2: Solution quality for fixed time budget for Bayesian D-Optimality. Results for the randomized algorithms are the average over $100$ independent runs. Note that the $x$-axis of both plots uses a logarithmic scale. For a limited number of calls to the oracle function, the \kas{} yields best performance. The \kas{} \textsc{(1)-(3)} use parameters $m = 1 + \lceil \sqrt{(p + 1)/2} \rceil , \varphi_1 = 1, \varphi_2 = 0.5$ and $\varepsilon = 0.7,0.5,0.3$ respectively. The \kas{} \textsc{(4)-(6)} use parameters $m = 1 , \varphi_1 = (p + 1)^{-1}, \varphi_2 = 1$ and, again, $\varepsilon = 0.7,0.5,0.3$. The \fastSGS{} uses parameters $\ell = \lfloor 2 + \sqrt{p + 1} \rfloor$ and $\varepsilon = 0.25$. \label{fig:2}
}
}
\end{figure*}
%
%
\section{Experimental Framework}
\label{experimental:framework}
In our set of experiments, we implement the \kas{} as describe in Algorithm \ref{alg:ksampl}. We always test our algorithm against these algorithms:
\begin{itemize}
    \item \textbf{\fantom{}}. This algorithm, which iterates a density greedy algorithm multiple times, is studied in \citet{DBLP:conf/wine/GuptaRST10} and \citet{DBLP:conf/icml/MirzasoleimanBK16}. 
    \item \textbf{\greedy{}}. This algorithm, studied in \cite{DBLP:conf/colt/FeldmanHK17}, consists of iterating a greedy algorithm multiple times. It uses Algorithm 1 in \cite{DBLP:journals/siamcomp/BuchbinderFNS15} as a sub-routine.
     \item \textbf{\fastSGS{}}. This algorithm is studied in \cite{DBLP:journals/corr/abs-2009-13998}, and it is essentially a fast implementation of the \simGreedy{} \cite{DBLP:journals/corr/abs-2009-13998}. This algorithm updates multiple solutions concurrently, and it picks the best of them.
    \item \textbf{\samplegreedy{}}. This algorithm is specifically designed to handle $p$-extendible systems (see \cite{DBLP:conf/colt/FeldmanHK17}). This algorithm samples points independently at random, and then it builds a greedy solution over the resulting set.
\end{itemize}
Note that these algorithms only require access to the independence oracle for the side constraints. In our experiments we do not consider algorithms that require access to the rank oracle, since they are impractical for our applications. We perform two sets of experiments, on the following applications:
\begin{itemize}
    \item \textbf{Video Summarization}. This problem asks to find a set of representative frames for a given video. We use Determinantal Point Processes to select a diverse set of frame. In order to get better summaries, we employ a face-recognition tool to identify faces in each segment. This experiment is described in Section \ref{section:DPP}, and the results are displayed in Figure \ref{fig:1}.
    \item \textbf{Bayesian D-Optimality}.
    Here, the goal is to design an experiment that maximizes the expected utility of the outcome, using preliminary observations. We use observations from the Berkeley Earth data-set to select thermal stations around the world, to measure the temperature with. This experiment is described in Section \ref{sec:bayesian_design}, and the results are displayed in Figure \ref{fig:2}.
\end{itemize}
The code and the datasets are available upon request.
%
%
\section{Video Summarization}
\label{section:DPP}
We study an application of our setting to a data summarization task: Given a video consisting of ordered frames, choose a subset of frames that gives a descriptive overview of the video. An effective way to select a diverse set of items is to apply Determinantal Point Processes \citep{macchi_1975}. For a thorough survey on Determinantal Point Processes and their applications, we refer the reader to \citet{DBLP:journals/ftml/KuleszaT12}. 

For a set of items $V = \{ 1, \dots, n \}$, a Determinantal Point Process (DPP) defines a discrete probability distribution over all subsets $S \subseteq V$ as $\pr{}{S} = \det_L(S)/\det(L + I)$, where $L \in \mathbb{R}^{n \times n}$ is a positive semidefinite matrix, $\det_L(S) \coloneqq \det((L_{i,j})_{i,j \in S})$ is the determinant of the sub-matrix of $L$ indexed by $S$ and $I$ is the $n \times n$-identity matrix. Intuitively, if $L$ expresses pairwise similarity, then the DPP prefers diversity.

In our setting, each item corresponds to a frame of a video segment. For each frame $i$, we compute a feature vector $\mathbf{f}_i$, consist of both visual features, such as color and SIFT features \citep{DBLP:conf/icml/KuleszaT11}, and qualitative information, such as size, colorfulness and luminosity. Following \citet{DBLP:conf/nips/GongCGS14}, we parameterize $L$ as $L_{i, j} \coloneqq z_i^{T} W^{T} W z_j$,
where $z_i \coloneqq \tanh (U \mathbf{f}_i)$. We then learn the parameters $U$ and $W$ using a neural network.

We select a representative summary by maximizing the function $\log \mbox{det}_L (S)$.
We impose the following additional side constraints. First, we impose an upper-bound on the maximum number of frames of each summary. Then, we partition each video into segments, and define a partition matroid to select at most $\ell_j$ frames in each segment $j$. Following \cite{DBLP:conf/aaai/MirzasoleimanJ018,DBLP:conf/nips/FeldmanK018}, we also use a face-recognition tool to identify actors in each movie, and select a summary containing at most $k_i$ frames showing face $i$. This additional constraint corresponds to a $p$-system $\ind{} = \{S \subseteq V \colon \absl{S \cap V_i} \leq k_i \}$, with $V_i$ all frames containing face $i$.\footnote{\label{note_p}The parameter $p$ is estimated by counting the total number of distinct faces $i$ that appear in more than $k_i$ frames.} In our experiments, the parameters $k_i$ are always set to a fixed constant for all videos. Hence, the only variable that affects $p$ is the total number of distinct faces in each movie.

For our experimental investigation, we use movies from the Frames Labeled In Cinema (FLIC) data-set \citep{modec13}. We consider all movies in this data-set with at least $200$ frames, as to highlight performance when dealing with large problem size.

The results are displayed in Figure \ref{fig:1}, where we describe the parameter choice for each algorithm. For each non-deterministic algorithm, results are the sample mean of $100$ independent runs. We observe that, for different parameter choice, our algorithm outperforms \fantom{} and the \fastSGS{}, and it has better adaptivity than the greedy algorithms. The solution quality for the \samplegreedy{} and \kas{}, with parameters as in Lemma \ref{thm:extendible}, is worse on these instances. 
%
%
%
%
\section{Bayesian D-Optimality}
\label{sec:bayesian_design}
Bayesian experimental design provides a general framework to select a set of experiments, that maximize the expected utility of the outcome. Formally, we want to estimate the parameter $\theta$ of a function $y = f_{\theta}(x) + w$, where $w$ is an error. In this framework, the input $x$ is generated by a set of experiments. Assuming that parameters are equipped with a prior, Bayesian optimality criteria are useful in identifying the right experiments to perform, in order to generate the input $x$. 

We focus on linear regressions of the form
$y = \theta^{T} X + w$,
with $y, w \in \mathbb{R}^n$, $\theta \in \mathbb{R}^m$ and $X \in \mathbb{R}^{m \times n}$. Furthermore, we assume independent and homoscedastic noise. We approach experimental design with the D-optimality criterion, although other methods can be used to this end \citep{DBLP:journals/jmlr/KrauseSG08}. This criterion consists of maximizing the determinant of the Fisher information matrix. As shown in \citet{Sebastiani_Winn}, for regressions as described above the D-optimality criterion is equivalent to maximizing the entropy.

We apply the Bayesian D-optimality criterion to the following setting. Consider a data-set consisting of monthly temperatures measured by thermal stations at different locations, over a period of time. We want to collect data to perform a regression for a model explaining the temperature variation from one measurements to the other one. Here, collecting temperatures with a single station corresponds to performing an experiment, and the goal is to identify appropriate stations to perform future measurements with.

Assuming independent and homoscedastic noise, we search for a feasible set of stations maximizing the entropy. Since temperature variation series follow a Gaussian process \citep{DBLP:journals/jmlr/KrauseSG08,DBLP:conf/aaai/000100QR19}, the entropy is defined as $\mathcal{H}(S) = \frac{1 + \ln (2\pi)}{2} \absl{S} + \frac{1}{2} \ln \mbox{det}_\Sigma (S)$, with $S $ a subset of stations, and $\Sigma $ the covariance matrix. The function $\mbox{det}_\Sigma (S)$ is the determinant of the covariance matrix corresponding to a set $S$ of stations. Note that the function $\mathcal{H}(S)$ is submodular and non-monotone. We consider an upper-bound on the solution size as a side constraint. Furthermore, we group stations that are located in the same geographical area, and we impose an upper-bound on the number of stations that can be chosen in each group. This additional constraint is useful when stations are not distributed uniformly across a territory (see \citet{DBLP:conf/aaai/000100QR19}). In our experiment, geographical areas correspond to continents. We remark that, if only a single cardinality constraint is given, then Bayesian optimality criteria can be optimized well via regularized Determinantal Point Processes \citep{DBLP:conf/aistats/DerezinskiLM20}. 

For our experiments we consider the Berkeley Earth climate data-set (http://berkeleyearth.org/data/). This data-set combines $1.6$ billion temperature reports from $16$ preexisting data archives, for over $39,000$ unique stations worldwide. We run all algorithms with parameters described as in Figure \ref{fig:2} for a fixed time budget. 

In Figure \ref{fig:2}, we report on the average solution quality achieved by each algorithm, after a fixed number of oracle calls. We observe that the \kas{} gets to a good solution more quickly that the other algorithms. 
All algorithms find similar solution qualities, for unlimited time budget, with the $\fantom{}$ and $\kas{}$. slightly outperforming the other algorithms.
%
%
\section{Conclusion}
\label{sec:conclusion}
In this paper, we develop the first algorithm for non-monotone submodular maximization under $p$-system and $p$-extendible system side constraints, with poly-logarithm adaptivity (see Lemma \ref{thm:specific} and Theorem \ref{thm:extendible}). This algorithm also competes with previous known results in terms of the query complexity and approximation guarantee (see Table \ref{table:results}).\\
We consider two applications and study the performance of our algorithm against several other algorithms suitable for this problem. We observe that our algorithms has superior adaptivity, and that it competes in terms of the query complexity (see Figure \ref{fig:1}-\ref{fig:2}).
\section{Acknowledgements}
\label{sec:acknowledgements}
We would like to thank Christopher Weyand for helping the authors to develop some of the code used for the experiments. We would like to thank Martin Schirneck for useful discussion on previous related work.
This research has been partly funded by the Federal Ministry of Education and Research of Germany in the framework of KI-LAB-ITSE (project number 01IS19066).

\bibliographystyle{plainnat}
\bibliography{bibliography}

\newpage
%
%
\onecolumn

\section*{Appendix}

%
%
\appendix
\section{Notation}
Here, we discuss the notation that will be used throughout the proofs in the Appendix, which is also discussed in the main body of the paper.\\
For any submodular function $f\colon 2^V \rightarrow \mathbb{R}$ and sets $S, U \subseteq V$, we define the marginal value of $S$ with respect to $U$ as $f (U \mid S) = f(S\cup U) - f(S)$. Note that, if $f$ only attains non-negative values, it holds that $f(U) \geq f(U \mid S)$ for all $S, U \subseteq V$.\\
We always use the notation introduced in Problem \ref{problem}, and we always denote with $\opt$ a solution to Problem \ref{problem}. Furthermore, we always denote with $n$ the size of the ground set $V$, i.e., $n$ is the number of singletons in our solution space. We denote with $r$ the maximum size of any feasible solution in $\ind{}$. This quantity is sometimes called the rank.\\
We denote with $m, \varepsilon$ the parameters as in Algorithm \ref{alg:ksampl}. Sets $\Omega_j$ and $\Lambda_j$  are as in Algorithm \ref{alg:ksampl}, and we denote with $V_j $ the ground set for the \asplusplus{} algorithm, during the $j$-th iteration of the for-loop lines \ref{for:kas_begin}-\ref{for:kas_end} of Algorithm \ref{alg:ksampl}. 
Furthermore $\delta_0, \delta$, and $\lambda$ are as in Algorithm \ref{alg:samplplusplus} and Algorithm \ref{alg:ksampl}. We denote with $n$ the problem size, and we denote with $r$ the maximum size of a feasible solution in $\ind{}$.

%
%
\section{Proof of Lemma \ref{lemma00}}
\label{appendix:A}
In this section, we prove a lemma that is useful to prove the desired approximation guarantee. Throughout the section, we always implicitly assume that parameters for Algorithm \ref{alg:samplplusplus} and Algorithm \ref{alg:ksampl} are as in Theorem \ref{thm1}, i.e., $\varphi_1 = 1$ and $\varphi_2 = 1/2$. The following lemma holds.
\begin{restatable}{lemma}{constantApprox}
\label{lemma00}
It holds $(p + 1)\expect{}{f(\Omega_j)} + \lambda r \expect{}{f(\Omega_j)}   \geq (1 - \varepsilon)^2 \expect{}{f(\Omega_j \cup (\opt \cap V_j)}$, for all sets $\Omega_j$.
\end{restatable}
On a high level, we prove that $\delta$ is an upper-bound for the best possible improvement up to a multiplicative constant, and that the marginal contribution of any point added to the current solution does not exceed $\delta$ in expected value, up to a multiplicative constant. We then combine this fact with the defining properties of the $p$-system to prove the claim, which holds for non-monotone functions. With this lemma, we prove that Algorithm \ref{alg:ksampl} yields, in expectation, a constant-factor approximation guarantee for Problem \ref{problem}. 
%
%
\subsection{Preliminary Results}
\label{subsecAppendix:preliminaryResults}
In order to prove Lemma \ref{lemma00}, we need the following technical proposition.
\begin{proposition}[Proposition 2.2 in \citet{inbook}]
\label{prop1}
Let $\{x_1, \dots, x_m\}, \{y_1, \dots, y_m\}$ be two sequences of non-negative real numbers. Suppose that it holds
\begin{equation*}
\sum_{j = 1}^i x_i \leq i,
\end{equation*}
for all $i \in [m]$ and $y_{i} \geq y_{i + 1}$ for all $i \in [m - 1]$. Then
\begin{equation*}
\sum_{j = 1}^m y_i \geq \sum_{j = 1}^m x_i y_i.
\end{equation*}
\end{proposition}
%
%
\subsection{Additional Lemmas}

We first prove the following result, which follows from Lemma 2 in \citet{DBLP:conf/stoc/BalkanskiRS19}. 
\begin{lemma}
\label{lemma2}
Fix an index $j$, and let $\delta, \Omega$ be as in Algorithm \ref{alg:samplplusplus}, as it runs over the set $V_j$. Then it holds
\begin{equation*}
\delta \geq (1 - \varepsilon ) \sup_{\{e \in V_j \setminus \Omega \colon \Omega \cup e \in \ind{}\}} f (e \mid \Omega).
\end{equation*}
\end{lemma}
\begin{proof}
We prove the claim by induction on the iterations of Algorithm \ref{alg:samplplusplus}. The base case trivially holds, due to the definition of $\delta$. Suppose that at some point the current solution $\Omega$ is updated to $\Omega \cup \{a_1, \dots, a_{\eta } \}$. We have that 
\begin{align*}
& \sup_{\{e \in V_{j}\setminus (\Omega \cup \{a_1, \dots, a_{\eta} \}) \colon \Omega \cup \{a_1, \dots, a_{\eta }\}\cup e \in \ind{} \}} f  ( e \mid \Omega \cup \{a_1, \dots, a_{\eta} \}) \\
& \leq \sup_{ \{e \in V_{j}\setminus (\Omega \cup \{a_1, \dots, a_{\eta} \}) \colon \Omega \cup \{a_1, \dots, a_{\eta }\}\cup e\in \ind{} \}} f (e \mid \Omega) \\
& \leq \sup_{ \{e \in V_{j}\setminus \Omega \colon \Omega \cup e \in \ind{} \}} f (e \mid \Omega) \\
& \leq \delta,
\end{align*}
where the first inequality uses the submodularity property of $f$; the second inequality holds since $\{e \in V_{j}\setminus (\Omega \cup \{a_1, \dots, a_{\eta} \}) \colon \Omega \cup \{a_1, \dots, a_{\eta }\}\cup e\in \ind{} \} \subseteq \{e \in V_{j}\setminus \Omega \colon \Omega \cup e \in \ind{} \}$; the last inequality follows due to the inductive hypothesis. 

We now show that the claim holds when $\delta$ is updated to $\delta' = (1 - \varepsilon )\delta$. At this point, it holds $X = \emptyset$, and each point $e^* \in V_j \setminus \Omega$ such that $\Omega \cup e^* \in \ind $ was discarded during a previous iteration. Denote with $\Omega'$ the solution at the iteration when $e^*$ was discarded, and let $\{a_1, \dots, a_{\eta}\}$ be the next set of points added to $\Omega'$. Since $e^*$ was discarded, then one of the following two conditions must hold:
\begin{enumerate}[align=left]
    \item $\Omega' \cup \{a_1, \dots, a_{\eta}\} \cup e^* \notin \ind{}$;
    \item $ f (e^* \mid \Omega' \cup \{a_1, \dots, a_{\eta}\} ) \leq \delta $.\label{last-item2}
\end{enumerate}
In the first case, due to the downward-closed property of $\ind{}$ it holds $\Omega \cup e^* \notin \ind{}$, which contradicts the definition of $e^*$. In the latter case it holds
\begin{equation*}
f(e^* \mid \Omega) \leq f (e^* \mid \Omega'\cup \{a_1, \dots, a_{\eta}\}) \leq \delta = \frac{\delta'}{1 - \varepsilon}
\end{equation*}
where the first inequality uses the fact that $f$ is submodular. The claim follows.
\end{proof}
%
%
We also need an additional lemma, to prove that Algorithm \ref{alg:ksampl} gives a constant-factor approximation for Problem \ref{problem}. 
%
%
\begin{lemma}
\label{lemma20}
Fix an index $j$, and let $\delta, X$ be as in Algorithm \ref{alg:samplplusplus}, as it optimizes the function $f$ over the set $V_j$. Denote with $\{a_i\}_i$ the points of $\Omega_j$,  sorted as they were added to it. It holds
\begin{equation*}
    \expect{a_i}{f(a_i \mid \{a_1, \dots , a_{i - 1} \})} \geq (1 - \varepsilon) \delta.
\end{equation*}
\end{lemma}
\begin{proof}
First, suppose that the constant $\delta$ is updated after the point $a_{i - 1}$ is added to the current solution. In that case, by definition of the set $X$, every point $e \in X \setminus \{a_1, \dots , a_{i - 1} \}$ such that $\{a_1, \dots , a_{i - 1} \} \cup e \in \ind $ yields $f(e \mid \{a_1, \dots , a_{i - 1} \}) > \delta$. Hence, the claim holds.

Suppose now that the constant $\delta$ is not updated after the point $a_{i - 1}$ is added to the current solution. Define the set $X_{i - 1}^{\ind{}} \coloneqq \{e \in X \colon \{a_1, \dots , a_{i - 1} \} \cup e \in \ind \}$. We first claim that the point $a_i$ is chosen uniformly at random over the set $X_{i - 1}^{\ind{}}$. In fact, if $a_i$ is added to the current solution, then there exists a set $\{x_1, \dots, x_{\eta}\}$ as in Line \ref{alg:sequence:alg} of Algorithm \ref{alg:rndseq} such that $a_i \in \{x_1, \dots, x_{\eta}\}$. Let $j \leq \eta$  be an index such that $a_i = x_j$. Due to the downward-closed property of $\ind{}$ it holds $\{e \colon A \cup \{x_1, \dots, x_{j}, e\} \in \ind{}\} \subseteq X \setminus \{x_1, \dots, x_{j}\}$. Hence,  $a_i$ is chosen uniformly at random among all points $e$ such that $A \cup \{x_1, \dots, x_{j - 1}, e\} = \{a_1, \dots, a_{j - 1}, e\} \in \ind$.

Then 
\[
\pr{a_i}{f(a_i \mid \{a_1, \dots , a_{i - 1} \}) > \delta } \geq \frac{\absl{X_{i - 1}}}{\absl{X_{i - 1}^{\ind{}}}} \geq \frac{\absl{X_{i - 1}}}{\absl{X}},
\]
where we have used that $X_{i - 1}^{\ind{}}\subseteq X_{i - 1}$ and that the point $e$ is chosen uniformly at random over the set $X_{i - 1}^{\ind{}}$. 

We now prove that $\absl{X_{i - 1}}/\absl{X} \geq (1 - \varepsilon)$. To this end, we first note that $\absl{X_{i}} \geq \absl{X_{i + 1}}$, for all indices $i$. Fix a point $e \in X_{i + 1}$. Then, this point yields $\{a_1, \dots, a_i \} \cup a \in \ind{}$ and $f(a \mid \{a_1, \dots, a_i \}) \geq \delta$. By the downward-closed property of $\ind{}$ we get $\{a_1, \dots, a_{i-1} \} \cup a \in \ind{}$, and by submodularity we get $f(a \mid \{a_1, \dots, a_{i-1} \}) \geq \delta$. Hence, $X_{i + 1} \subseteq X_i$ and $\absl{X_{i}} \geq \absl{X_{i + 1}}$ as claimed. It follows that the \binarysearch{} sub-routine of Algorithm \ref{alg:samplplusplus} terminates whenever $\eta = \min \{ i \colon \absl{X_i} \leq (1 - \varepsilon ) \absl{X} \}$, which implies that $\absl{X_{i - 1}} > (1 - \varepsilon ) \absl{X}$. The claim follows since it holds $\expect{a_i}{f(a_i \mid \{a_1, \dots , a_{i - 1} \})} \geq \pr{a_i}{f(a_i \mid \{a_1, \dots , a_{i - 1} \}) > \delta } \delta $.
\end{proof}
\subsection{Proof of Lemma \ref{lemma00}}
\label{sec:proof:thm1}
Using Proposition \ref{prop1}, Lemma \ref{lemma2}-\ref{lemma20}, we can prove Lemma \ref{lemma00}. This proof uses ideas in \citet{DBLP:conf/wine/GuptaRST10}.
\begin{proof}[Proof of Lemma \ref{lemma00}]
Denote with $\{a_{i}\}_i$ the points of $\Omega_j$ in the order that they were added to $\Omega_j$. Define the set 
\[
W_{\delta_0} \coloneqq \{e \in \opt \cap V_j \colon f(e \mid \Omega_j) \geq \delta_0 \}.
\]
Note that this set consists of all points of $\opt \cap V_j$ such that their marginal contribution is above $\delta_0$, when added to the current solution. First, fix a set $\Omega_j$, and suppose that $f (a_i \mid \{a_1, \dots, a_{i - 1}\}) \geq \delta$ for all indices $i$. Define the sets $A_{i} \coloneqq \left \{ e \in W_{\delta_0} \setminus \{a_1, \dots, a_i \} \colon \{a_1, \dots, a_i \}\cup e \in \ind{} \right \}$. Since the system $\ind{}$ is downward-closed, then $A_{i} \subseteq A_{i - 1}$. Define the sets $D_{i} \coloneqq A_{i - 1} \setminus A_{i}$. Note that these sets consist of all points in $W_{\delta_0}$ that yield a feasible solution when added to $\{a_1, \dots, a_{i - 1} \}$, but that violate side constraints when added to $\{a_1, \dots, a_i \}$.

We now claim that the set $\left \{a_1, \dots, a_i \right \}$ is a maximal independent set for
\begin{equation*}
\left \{a_1, \dots, a_i \right \} \cup \left ( D_1 \cup \dots \cup D_i \right ) = \left \{a_1, \dots, a_i \right \} \cup \left ( W_{\delta_0} \setminus A_i \right ). 
\end{equation*}
To this end, note that the set $\left \{a_1, \dots, a_i \right \}$ is independent by definition, and that any point $e\in \left ( W_{\delta_0} \setminus A_i \right ) \setminus \left \{a_1, \dots, a_i \right \}$ is such that $\left \{a_1, \dots, a_i \right \} \cup e \notin \ind{}$. Hence $\left \{a_1, \dots, a_i \right \}$ is maximal as claimed. Note also that $D_1 \cup \dots \cup D_i \subseteq W_{\delta_0}$ is an independent set, due to the subset-closure of $\ind{}$. Since $\ind{}$ is a $p$-system, then it holds
\begin{equation}
\label{eq1:claim2:thm1}
    \absl{D_1} + \dots + \absl{D_i } = \absl{D_1 \cup \dots \cup D_i } \leq p \absl{\left \{a_1, \dots, a_i \right \}} = pi.
\end{equation}
Furthermore, using submodularity and Lemma \ref{lemma2} it holds
\begin{align*}
\absl{D_i } \delta & \geq (1 - \varepsilon) \absl{D_i } \mbox{sup}_e f (e \mid \left \{a_1, \dots , a_{i - 1} \right \})\\
& \geq (1 - \varepsilon) f(D_i  \mid \left \{a_1, \dots , a_{i - 1} \right \})\\
& \geq (1 - \varepsilon) f \left (D_i \mid \Omega_j\right ),
\end{align*}
where the last inequality follows from submodularity. Combining this with \eqref{eq1:claim2:thm1} and Proposition \ref{prop1}, we get
\begin{equation*}
    p\sum_i f(a_i \mid \left \{a_1, \dots , a_{i - 1} \right \}) \geq \sum_i \absl{D_i } \delta \geq (1 - \varepsilon) \sum_i f \left (D_i \mid \Omega_j \right ) \geq f\left (W_{\delta_0} \mid \Omega_j  \right ),
\end{equation*}
where the last inequality uses submodularity. Hence, rearranging yields 
$p f(\Omega_j) \geq f(\Omega_j \cup W_{\delta_0}) - f(\Omega_j)  =  f( W_{\delta_0} \mid \Omega_j)$. If we unfix the set $\Omega_j$ and take the expected value, and using Lemma \ref{lemma20} we get 
\begin{equation}
\label{eq:new_one}
p \expect{}{f(\Omega_j)} \geq (1 - \varepsilon)^2 \expect{}{f(W_{\delta_0} \mid \Omega_j)}.
\end{equation}

Using Lemma \ref{lemma20} again, we have that $\expect{a_i}{f(a_i \mid \{a_1, \dots, a_{i - 1} \})} \geq 0$, for all points $a_i$ added to $\Omega_j$. It follows that $\expect{}{f(\Omega_j)} \geq \expect{}{f(a_0)}$, with $a_0$ the first point added to $\Omega_j$. Then, from the definition of $\delta_0$ it follows that $\lambda \expect{}{f(\Omega_j)} \geq \lambda \expect{}{f(a_0)} = \expect{}{\delta_0}$. Hence, using submodularity and the linearity of the expected value, we get 
\begin{equation}
\label{eq:new_two}
\lambda r \expect{}{f(\Omega_j)} \geq r \expect{}{\delta_0} \geq \expect{}{\sum_{e \in (\opt \cap V_j)\setminus W_{\delta_0}} f(e \mid \Omega_j)} \geq \expect{}{f((\opt \cap V_j)\setminus W_{\delta_0} \mid \Omega_j)},
\end{equation}
where we have used submodularity.

Combining \eqref{eq:new_one} with \eqref{eq:new_two} and using submodularity again we get $p \expect{}{f(\Omega_j)} +  \lambda r \expect{}{f(\Omega_j)} \geq (1 - \varepsilon)^2 \expect{}{f(\opt \cap V_j \mid \Omega_j)}$. The claim follows by rearranging.
\end{proof}
%
%
\section{Proof of Theorem \ref{lemma00}}
\label{appendix:B1}
%
%
\subsection{Preliminary results.} 
In our analysis we consider the following well-known result.
\begin{lemma}[Theorem 2.1 in \citet{DBLP:journals/siamcomp/FeigeMV11}]
\label{thm23}
Let $U \subseteq V$ be a set chosen uniformly at random. Then it holds $\expect{}{f(U)} \geq f(O)/4$, with $O \subseteq V$ the subset attaining the maximum $f$-value.
\end{lemma}
Furthermore, we also consider the following properties of submodular functions.
\begin{lemma}[Lemma 10 in \citet{DBLP:conf/colt/FeldmanHK17}]
\label{thm2}
For any fixed $m$-tuple of mutually disjoint sets $\Omega_j$ it holds $(m - 1)f(\opt) \leq \sum_{j \leq m} f(\Omega_j \cup \opt)$.
\end{lemma}

\begin{lemma}[Lemma 11 in \citet{DBLP:conf/colt/FeldmanHK17}]
\label{thm22}
Let $f \colon 2^V \rightarrow \mathbb{R}_{\geq 0}$ be a non-negative submodular function. For every three sets $A, B, C \subseteq V$ it holds $f(A \cup (B \cap C)) + f(B\setminus C) \geq f(A \cup B)$.
\end{lemma}
\subsection{Proof of Theorem \ref{thm1}}
Using Lemma \ref{thm23}-\ref{thm22} together with Lemma \ref{lemma00}, we can prove Theorem \ref{thm1}. 
\begin{proof}[Proof of Theorem \ref{thm1}] Fix an $m$-tuple of sets $\Omega_j$ for Algorithm \ref{alg:ksampl}, and consider the sets $\opt \setminus \Omega_j$, for all $j \leq m$. Note that it holds $V_j = V \setminus \cup_{i \leq j} \Omega_i$. Hence, $f(\opt \setminus V_i) = f(\opt \setminus (V \setminus \cup_{i \leq j} \Omega_i)) = f(\cup_{i \leq j} (\opt \cap \Omega_j ) ) \leq \sum_{i \leq j} f(\opt \cap \Omega_j)$, where the last inequality uses submodularity.

Unfixing the sets $\Omega_j$ and taking the expected value yields 
\begin{equation}
\label{eq:thm123}
    \expect{}{f(\opt \setminus V_i)} \leq \sum_{i \leq j} \expect{}{f(\opt \cap \Omega_j)},
\end{equation}
for all indices $j \leq m$. We have that it holds
\begin{align*}
    (m - 1) f(\opt) & \leq \sum_{j \leq m} \expect{}{f(\opt \cup \Omega_j)} \\
    & \leq \sum_{j \leq m}\expect{}{f(\Omega_j \cup (\opt \cap V_j))} + \sum_{j \leq m} \expect{}{f(\opt \setminus V_j)}\\
    & \leq m \frac{p + 1}{(1 - \varepsilon)^2} \expect{}{f(\Omega_j)} + m \frac{\lambda r}{(1 - \varepsilon)^2} \expect{}{f(\Omega_j)} + \sum_{j \leq m} \sum_{i \leq j} \expect{}{f(\opt \cap \Omega_j)}\\ 
    & \leq m (p + 1) \frac{1 + \varepsilon}{(1 - \varepsilon)^2} \expect{}{f(\Omega_j)} + \sum_{j \leq m} \sum_{i \leq j} \expect{}{f(\opt \cap \Omega_j)}\\ 
    & \leq m (p + 1) \frac{1 + \varepsilon}{(1 - \varepsilon)^2} \expect{}{f(\Omega_j)} + 4 \sum_{j \leq m} \sum_{i \leq j} \expect{}{f(\Lambda_j)},\\ 
    & \leq m (p + 1) \frac{1 + \varepsilon}{(1 - \varepsilon)^2} \expect{}{f(\Omega^*)} + 2 m (m - 1) \expect{}{f(\Omega^*)},
\end{align*}
where the first inequality follows from Lemma \ref{thm2}; the second inequality follows from Lemma \ref{thm22}; the third inequality follows from \eqref{eq:thm123} and Lemma~\ref{lemma00}; the fourth inequality follows from Lemma \ref{thm23}; the last inequality follows since $f(\Omega^*)$ is maximum over $f(\Omega_j)$ and $f(\Lambda_j)$. The claim follows by rearranging the inequality above.
\end{proof}
%
%
\section{Proof of Lemma \ref{lemma:runTime}}
\label{appendix:C}
We now prove an upper-bound on the run time for Algorithm \ref{alg:ksampl}.
\runTime*
\begin{proof}
First, note that Algorithm \ref{alg:rndseq} requires no function evaluation, and it always return a sequence $\{a_i\}_i$ of length at most $r$. 

At each step of Algorithm \ref{alg:samplplusplus}, the $\binarysearch$ sub-routine requires $\bigo{\log (r)}$ iterations, sine the set $J$ has size at most $r$. Each iteration of this sub-routine requires $\bigo{1}$ rounds of adaptivity, and $\bigo{\log (r)}$ function evaluations.

Note also that the while-loop, lines \ref{alg:sampl:beginWhile2}-\ref{alg:sampl:endWhile2} of Algorithm \ref{alg:samplplusplus} terminates after at most $\bigo{ \varepsilon^{-1}\log n }$ iterations. In fact, we have that $\absl{X} \leq n$, and that at each iteration the size of the new set $X$ decreases of a multiplicative factor of $(1 - \varepsilon)$. Similarly, the outer while-loop, lines \ref{alg:samplplusplus:beginWhile1}-\ref{alg:samplplusplus:endWhile1} of Algorithm \ref{alg:samplplusplus} terminates after at most $\bigo{ \varepsilon^{-1}\log (r/p\varepsilon) }$ iterations.

Hence, each call of Algorithm \ref{alg:samplplusplus} requires $\bigo{m \varepsilon^{-2} \log(r/(p\varepsilon ))\log (\varepsilon^{-1} \log r) \log n)}$ rounds of adaptivity. Similarly, since the $\binarysearch $ sub-routine requires $\bigo{n \log (r)}$ function evaluations, then the query complexity is $\bigo{m n \varepsilon^{-2} \log(r/(p\varepsilon ))\log (r) \log n)}$.
\end{proof}
%
%
\section{Proof of Lemma \ref{thm:specific}}
\label{appendix:E}
We perform the run time analysis for an optimal choice of the parameter $m$. We have that the following theorem holds.
\specific*
\begin{proof}
We start with the approximation guarantee. Denote with $\Omega^*$ an approximate solution found by Algorithm \ref{alg:ksampl}, and let $\opt$ be the optimal solution for Problem \ref{problem}. Then, from Theorem \ref{thm1} we get
\begin{equation*}
    f(\opt) \leq m \left ( \frac{(1 + \varepsilon) (p + 1)}{(1 - \varepsilon)^2 (m - 1)} + 2 \right ) \expect{}{f(\Omega^*)}
\end{equation*}
Substituting $m = 1 + \lceil \sqrt{(p + 1)/2} \rceil$ and rearranging yields
\begin{align*}
    f(\opt)  & \leq \frac{1 + \varepsilon}{(1 - \varepsilon)^2}\left ( p + 2 \left \lceil \sqrt{\frac{p + 1}{2}} \right \rceil + (p + 1) \left \lceil \sqrt{\frac{p + 1}{2}} \right \rceil^{-1} + 3 \right ) \expect{}{f(\Omega^*)} \\
    & \leq \frac{1 + \varepsilon}{(1 - \varepsilon)^2} \left ( p + 2 \left ( \sqrt{\frac{p + 1}{2}} + 1 \right ) + (p + 1) \left ( \sqrt{\frac{p + 1}{2}} \right )^{-1} + 3\right ) \expect{}{f(\Omega^*)} \\ 
    & = \frac{(1 + \varepsilon)(p  + 2 \sqrt{2 ( p + 1)} + 5)}{(1 - \varepsilon)^2} \expect{}{f(\Omega^*)}.
\end{align*}
Hence, the claim on the approximation guarantee follows. The upper-bounds on the adaptivity, depth, and total number of calls to the valuation oracle function follow directly from Lemma \ref{lemma:runTime}.
\end{proof}
%
%
\section{Proof of Lemma \ref{lemma:indOracle}}
\label{appendix:D}
In this section, we perform the run time analysis for Algorithm \ref{alg:ksampl} with respect to the calls to the independence oracle for the $p$-system. The following lemma holds.
\indOracle*
In order to prove this Lemma, we use the following well-known result.
\begin{theorem}[Theorem 6 in \citet{DBLP:journals/jcss/KarpUW88}]
\label{thm:Karp}
Algorithm \ref{alg:rndseq} terminates after $\mathcal{O}(\sqrt{r})$ steps. 
\end{theorem}
Combining Theorem \ref{thm:Karp} with Lemma \ref{lemma:runTime}, we prove Lemma \ref{lemma:indOracle} as follows.
\begin{proof}
We first observed that the independence oracle for the $p$-system is called by Algorithm \ref{alg:rndseq}, and also by Algorithm \ref{alg:samplplusplus}. 

Since at each iteration of Algorithm \ref{alg:rndseq}, queries to the oracle for the $p$-system are independent, the from Theorem \ref{thm:Karp} it follows that Algorithm \ref{alg:rndseq} requires $\bigo{\sqrt{r}}$ rounds of independent calls to the oracle for the $p$-system. Furthermore, all calls to the independence oracle for the $p$-system in Algorithm \ref{alg:samplplusplus} are independent. Combining these observations with Lemma \ref{lemma:runTime} it follows that Algorithm \ref{alg:ksampl} requires $\bigo{m\sqrt{n} /\varepsilon^2 \log (n) \log (r/\varepsilon) }$ rounds of independent calls to the oracle for the $p$-system.

The claim follows, since at each round at most $\bigo{n}$ calls to the oracle for the $p$-system are executed in parallel.
\end{proof}
%
%
\section{Proof of Theorem \ref{thm:extendible}}
\label{appendix:G}
\subsection{Preliminary Results.} 
In order to prove Theorem \ref{thm:extendible} we use the following well-known result.
\begin{lemma}[Lemma 2.2 in \citet{DBLP:conf/soda/BuchbinderFNS14}]
\label{thm23sdf}
Let $\Omega \subseteq V$ be a set such that each element appears in $\Omega$ with probability at most $k$. Then it holds $\expect{}{f(U)} \geq (1 - k) f(\emptyset )$.
\end{lemma}
\subsection{Additional Lemmas.}
In this section, we prove the following theorem.
\mainthmext*
In order to prove this theorem, we introduce additional notation. First of all, since $m = 1$, we need not specify the index on the input search space $V_i$ and solution $\Omega_i$ of Algorithm \ref{alg:samplplusplus}, and we simply use the notation $V = V_1$ and $\Omega = \Omega_1$. Again we define $\absl{V} = n$. Furthermore, we define an ordering of the points $\{v_i \}_i = V$, with $v_i$ the $i$-th point sampled by Algorithm \ref{alg:samplplusplus} during run time. All points of $V$ that are not sampled during run time, are placed at the end of the sequence $\{v_i \}_i$ in random order. Furthermore, we define the sets $T_{i} = \{v_1, \dots, v_{i} \} \cap \Omega$, and we define the sequence $\{v^*_i\}$ as
\begin{equation*}
    v^*_i \coloneqq \max_{ \{ v \in V\setminus T_{i - 1} \colon T_{i - 1}  \cup v \in \ind{} \} } f(v \mid T_{i - 1}).
\end{equation*}
For each point $v \in V$, denote with $\mathcal{X}_v$ an indicator function such that $\mathcal{X}_v = 1$ if $v$ is sampled as part of a random feasible sequence $\{a_1, \dots, a_{\eta}\}$, and $\mathcal{X}_v = 0$ otherwise. We also consider a sequence $\{O_i\}_{i = 0}^n$ of sets $O_i \subseteq V$ defined recursively as
\begin{itemize}
    \item $O_0 \coloneqq \emptyset$;
    \item if $v_i \in \Omega$, then $O_i \subseteq \opt \setminus (T_{i - 1} \cup_{j = 0}^{i - 1} O_j)$ is a set of minimum size such that $(\opt  \setminus (\cup_{j = 0}^{i} O_j) )\cup (T_{i - 1}\cup v_i ) \in \ind{}$;
    \item if $v_i \notin \Omega$, $\mathcal{X}_{v_i} = 1$, and $v_i \in \opt \setminus (\cup_{j = 0}^{i - 1} O_j)$, then $O_i = \{v_i\}$;
    \item if $v_i \notin \Omega$ and $v_i \notin \opt \setminus (\cup_{j = 0}^{i - 1} O_j)$, or if $\mathcal{X}_{v_i} = 0$, then $O_i = \emptyset $.
\end{itemize}
Finally, we define the set $O \coloneqq (\opt  \setminus (\cup_{i = 0}^n O_i) )\cup T_n = (\opt  \setminus (\cup_{i = 0}^n O_i) )\cup \Omega$.  

Following this notation, we first prove the following lemma.
\begin{lemma}
\label{lemma:p_ext1}
Fix all random decisions of Algorithm \ref{alg:samplplusplus}. Then it holds
\begin{equation*}
    f(\Omega) + \absl{O\setminus \Omega}\delta_0 \geq f(\Omega \cup \opt ) - \sum_{i = 0}^n \absl{O_i \setminus \Omega}f(v^*_i \mid T_{i - 1} \}).
\end{equation*}
\end{lemma}
\begin{proof}
First, we prove that it holds
\begin{equation}
\label{eq:p_ext1}
f(S) + \absl{O\setminus S}\delta_0 \geq f(O).
\end{equation}
To this end, note that since $(\opt  \setminus (\cup_{i = 0}^n O_i) )\cup \Omega \in \ind{}$, then it holds $\{v\} \cup \Omega \in \ind{}$ for all $v \in O\setminus \Omega$. Hence, by the termination criterion of Algorithm \ref{alg:samplplusplus}, we have that $f(v \mid \Omega ) \leq \delta_0$. 
\begin{equation*}
    f(O) \leq f(\Omega ) + \sum_{v \in O \setminus \Omega} f(v \mid \Omega) \leq f(\Omega) + \absl{O\setminus \Omega}\delta_0,
\end{equation*}
where we have used submodularity. Then \eqref{eq:p_ext1} follows.

Next, we prove that it holds
\begin{equation}
\label{eq:p_ext2}
    f(O) \geq f(\Omega \cup \opt ) - \sum_{i = 1}^n \absl{O_i \setminus \Omega }f(v_i^* \mid T_{i - 1}).
\end{equation}
Note that the claim of this lemma follows by combining \eqref{eq:p_ext1} and \eqref{eq:p_ext2}.

To prove \eqref{eq:p_ext2}, we first observe that the sets $O_i \setminus \Omega $ are mutually disjoint, and we can write $O = (\Omega \cup \opt ) \setminus (\cup_{i=1}^n (O_i \setminus \Omega))$. Using this equality, we have that
\begin{align*}
    f(O) & = f(\Omega \cup \opt ) - \sum_{i = 1}^n f(O_i \setminus \Omega \mid (\Omega \cup \opt ) \setminus (\cup_{j = 0}^{i - 1} (O_j \setminus \Omega))) \\
    & \geq f(\Omega \cup \opt ) - \sum_{i = 1}^n f(O_i \setminus \Omega \mid T_{i - 1} ) \\    
    & \geq f(\Omega \cup \opt ) - \sum_{i = 1}^n \sum_{v \in O \setminus \Omega} f(v \mid T_{i - 1} ),
\end{align*}
where the first equation is the telescopic sum, the second one uses submodularity, together with the fact that $T_{i - 1} \subseteq \Omega$, and the third one use submodularity again. Then \eqref{eq:p_ext2} follows from the definition of $v^*_i$.
\end{proof}
Next, using Lemma \ref{lemma:p_ext1} we prove the following result.
\begin{lemma}
\label{lemma:p_ext2}
It holds
\begin{equation*}
    \expect{v_i}{\absl{O_i \setminus \Omega } f(v^*_i \mid T_{i - 1} ) } \leq (1 -  \varepsilon )^2 \frac{p}{p + 1}\expect{v_i}{\mathcal{X}_{v_i}f(v_i \mid T_{i - 1} )}.
\end{equation*}
\end{lemma}
\begin{proof}
We first observe that if $\mathcal{X}_i = 0$, then the claim holds since $\absl{O_i \setminus \Omega} = \emptyset$. Hence, we prove the claim by conditioning on the event $\{\mathcal{X}_{v_i} = 1 \}$. In this case, the point $v_i$ is added to $\Omega$ w.p. $(p + 1)^{-1}$. 

If $v_i$ is not added to the current solution, then $\absl{O_i} \leq 1$, and conditioning on the event $\{v_i \notin \Omega \}$ we get
\begin{align*}
    \expect{a_i}{\absl{O_i \setminus \Omega } f(v^*_i \mid T_{i - 1} ) \mid \mathcal{X}_{v_i} = 1, v_i \notin \Omega } & \leq  \expect{a_i}{\mathcal{X}_{v_i}f(v^*_i \mid T_{i - 1} ) \mid \mathcal{X}_{v_i} = 1, v_i \notin \Omega } \\
    & \leq \frac{p}{p + 1} \expect{a_i}{\mathcal{X}_{v_i} f(v_i^* \mid T_{i - 1}) \mid \mathcal{X}_{v_i} = 1},
\end{align*}
where we have used that $\pr{}{a_i \notin S} = 1 - (p + 1)^{-1}$. 
Furthermore, if the solution $v_i $ is added to $\Omega$, then the set $O_i$ has size at most $\absl{O_i} \leq p $, since $\ind{}$ is a $p$-extendible system. Hence,
\begin{align*}
    \expect{a_i}{\absl{O_i \setminus \Omega } f(v^*_i \mid T_{i - 1} ) \mid \mathcal{X}_{v_i} = 1, v_i \in \Omega } & \leq  p \expect{a_i}{\mathcal{X}_{v_i}f(v^*_i \mid T_{i - 1} ) \mid \mathcal{X}_{v_i} = 1, v_i \in \Omega } \\
    & \leq \frac{p}{p + 1} \expect{a_i}{\mathcal{X}_{v_i} f(v_i^* \mid T_{i - 1}) \mid \mathcal{X}_{v_i} = 1}.
\end{align*}
The claim follows combining the two chains of inequalities above, together with Lemma \ref{lemma2} and Lemma \ref{lemma20}.
\end{proof}
We also need the following lemma, to prove the main theorem. 
\begin{lemma}
\label{lemma:p_ext3}
It holds
\begin{equation*}
    \sum_{i = 1}^n \expect{}{\absl{O_i \setminus \Omega } f(v^*_i \mid T_{i - 1} ) } \leq p\expect{}{f(\Omega)}.
\end{equation*}
\end{lemma}
\begin{proof}
For each $v_i \in V$, let $\mathcal{G}_{v_i}$ be a random variable whose value is equal to in the increase in the value of the current solution when added to it. Note that if $v_i$ yields $\mathcal{X}_{v_i} = 0$, then $\mathcal{G}_{v_i} = 0$ because it cannot be added to the current solution. Then,
\begin{equation}
\label{eq:asd_lkj1}
    \expect{v_i}{\mathcal{G}_{v_i}} = \pr{}{v_i \in \Omega} \expect{v_i}{\mathcal{X}_{v_i} f(v_i \mid T_{i - 1} )} = \frac{1}{p + 1} \expect{a_i}{\mathcal{X}_{v_i} f(v_i \mid T_{i - 1})}.
\end{equation}
The claim follows using Lemma \ref{lemma:p_ext2}, and using the law of total probability and linearity of the expected value.
\end{proof}
\subsection{Proof of Theorem \ref{thm:extendible}.}
We now have all necessary tools to prove the main theorem.
\begin{proof}[Proof of Theorem \ref{thm:extendible}]
Combining Lemma \ref{lemma:p_ext1}, taking the expected value, and combining with Lemma \ref{lemma:p_ext3} we get
\begin{equation}
\label{eq:asdlkj2}
    \frac{p + 1}{(1 - \varepsilon)^2} \expect{}{f(\Omega)} + \expect{}{\absl{O\setminus \Omega}\delta_0} \geq \expect{}{f(\Omega \cup \opt )}.
\end{equation}
Denote with $\{a_i\}_i$ the points of $\Omega $ in the order that they were added to $\Omega$. From Lemma \ref{lemma20} and submodularity, we have that $\expect{a_i}{f(a_i \mid \{a_1, \dots, a_{i - 1} \})} \geq 0$, for all points $a_i$ added to $\Omega_j$. It follows that $\expect{}{f(\Omega_j)} \geq \expect{}{f(a_0)}$, with $a_0$ the first point added to $\Omega_j$. Hence, from the definition of $\delta_0$, and since $\absl{O \setminus \Omega} \leq r$ due to feasibility, we get $\varepsilon \expect{}{f(\Omega)} \geq \expect{}{\absl{O\setminus \Omega}\delta_0}$. Substituting in \eqref{eq:asdlkj2} we get 
\begin{equation*}
    \frac{(p + 1)(1 + \varepsilon)}{(1 - \varepsilon)^2} \expect{}{f(\Omega)} \geq \expect{}{f(\Omega \cup \opt )}.
\end{equation*}
To conclude the proof, we observe that the function $g(S) = f(S \cup \opt)$ is a submodular function. Since each element of $V$ appears in $\Omega$ w.p. at most $(p + 1)^{-1}$, then by Lemma \ref{thm23sdf} we get \begin{equation*}
\expect{}{f(\Omega \cup \opt )} = \expect{}{g(\Omega)} \geq \frac{p}{p + 1} g(\emptyset ) = \frac{p}{p + 1} f(\opt ).
\end{equation*}
The claim follows.
\end{proof}
\end{document}